\documentclass[preprint,12pt]{elsarticle}

\usepackage{graphics}

\usepackage{amssymb}
\usepackage{amsthm}
\usepackage{mathrsfs}
\usepackage{multicol}

\usepackage{graphicx}
\usepackage{rotating,booktabs}
\usepackage{dcolumn}
\usepackage{tabularx}

\biboptions{square,sort&compress}

\journal{ }

\begin{document}

\newtheorem{proposition}{Proposition}
\newtheorem{remark}{Remark}
\newcolumntype{z}[1]{D{.}{.}{#1}}

\begin{frontmatter}



 \title{Probabilities of Positive Returns and  Values  of Call Options }



\author[label1]{Guanghui Huang\corref{cor1}}
  \ead{hgh@cqu.edu.cn}
\author[label2]{Jianping Wan}
 \cortext[cor1]{Corresponding author.}
\address[label1]{Department of Statistics and Actuarial Science, Chongqing
University, Chongqing, 400030,  P.R. CHINA;}
\address[label2]{Department of Mathematics, Huazhong University of Science \& Technology, Wuhan, 430074,   P.R.CHINA.}

\begin{abstract}
The true probability of a European call option to achieve positive
return is investigated under the Black-Scholes model. It is found
that the probability is determined by those market factors appearing
in the BS formula, besides the growth rate of stock price. Our
numerical investigations indicate that the biases of BS formula is
correlated with the growth rate of stock price. An alternative
method to price European call option is proposed, which adopts an
equilibrium argument to determine option price through the
probability of positive return. It is found that the BS values are
on average larger than the values of proposed method for
out-of-the-money options, and smaller than the values of proposed
method for in-the-money options. A typical smile shape of implied
volatility is also observed in our numerical investigation. These
theoretical observations are similar to the empirical anomalies of
BS values, which indicates that the proposed valuation method may
have some merit.
\end{abstract}

\begin{keyword}
Black-Scholes formula \sep
 probability of positive return \sep
  growth
rate of stock price \sep
 equilibrium option pricing.


\end{keyword}

\end{frontmatter}



\section{Introduction}
The departure of market prices from their theoretical Black-Scholes
(BS) values has been discussed for a long time, and many market
factors have been used to explain those anomalies, such as
volatility, interest rate, moneyness, time to expiration,
transaction costs, market liquidation, trading volume, bid-ask
spread, option open interest, and short sale constraints, etc.
Examples of these studies include \cite{chiras_manaster_1978,
macbeth_merville_1979,rubinstein_1985,longstaff_1995,
thompson_williams_1999,kuwahara_marsh_2000,isaenko_2007},
 and the references therein.
Although the growth rate is an important factor to describe the
dynamics of stock prices, it is rarely discussed in the literature
for the biases of BS formula, except the work of
\cite{thompson_williams_1999}. The purposes of this paper are to
identify whether the growth rate  of stock price can be used to
explain the biases of BS formula, and to determine the price of
European call option through an equilibrium argument under the BS
model.

Risk neural valuation approach was first introduced by
\cite{black_scholes_1973}, and has been extended by
\cite{merton_1973,cox_ross_1976,harrison_kreps_1979,
harrison_pliska_1981}, and others. The risk neural price of an
option  is  the discounted expectation of its payoff under a risk
neural pricing measure, which satisfies the martingale constraint
that  the  discounted process of stock price is a martingale under
this measure. The existence of risk neural martingale measure is
equivalent to the absence of arbitrage opportunity.
\cite{longstaff_1995} investigated the existence of martingale
restriction using S\&P 100 option prices, and found that the data
strongly reject the martingale restriction. Those observed anomalies
of BS formula indicate that options should be valued by equilibrium
methods rather than no-arbitrage models \cite{longstaff_1995}.

It is the physical measure of the market that determines the
dynamics of stock price, rather than the artificial equivalent
martingale measure, therefore the true probability of an option to
bring a positive return to the holder is calculated under the
physical measure. The probability is determined by all of the market
factors appearing in the BS formula, besides the growth rate of
stock price.

Our numerical investigations show that when the growth rates are
negative, there is almost no option with probability of positive
return larger than $50\%$ with respect to their BS  prices. When the
growth rates are positive, the probabilities are on average larger
than $50\%$ for   in-the-money options, and less than $50\%$ for
out-of-the-money options. Therefore the equilibrium prices should be
adjusted from their BS values, such that the market prices are on
average larger than the BS values for in-the-money options, and less
than the BS values for out-of-the-money options.  These theoretical
predictions are similar to the empirical results reported by
\cite{macbeth_merville_1979}, where the market prices are on average
lager than the BS values for in-the-money options, and less than the
BS values for out-of-money options.

A equilibrium argument is applied to determine the option price
through the probabilities of positive returns. The market clearing
price is the price at which both parties in the transaction come to
an agreement on the probability of positive return. The individual
requirements for probabilities of positive returns can be different
among the investors according to their attitudes to risk, and the
option price can be calculated for each specified probability of
positive return.

The performance of the proposed method and the BS formula are
compared numerically under the BS model. It is found that the values
of proposed method  are on average larger than the BS values for
  in-the-money options, and less
than the BS values  for  out-of-the-money options. And a typical
smile shape of implied volatilities is also observed for different
growth rates of stock prices. These theoretical phenomenon are
similar to the empirical anomalies of
 BS values  reported by  \cite{macbeth_merville_1979,kuwahara_marsh_2000,isaenko_2007}, which
indicate that the proposed valuation method may have some merit.

The rest of this article is organized as follows. Section 2 derives
the probabilities of positive returns for European call options
under the BS model, and investigates the influence of  market
factors on the probabilities numerically. Section 3 determines the
option price using probability of positive return. Section 4
compares  BS formula with the proposed method from their values and
the implied volatilities respectively. Section 5 summarizes this
paper and discusses the results.

\section{Probabilities of Positive Returns}
\subsection{The Black-Scholes Model}
In the Black-Scholes framework, the market is frictionless, and the
dynamics of stock price is described by a geometric Brownian motion,
\begin{equation}\label{price}
\frac{\mathrm{d}S_t}{S_t} = \mu \mathrm{d}t +  \sigma \mathrm{d}W_t,
\end{equation}
where $S_t$ is the stock price at time $t$,
 $\mu$ is the expected growth rate of stock price, $\sigma$ is the
volatility, $W_t$ is a standard Brownian motion. The solution of
equation (\ref{price}) is
\begin{equation}
S_t = S_0 \exp\left( \sigma W_t + \left( \mu - \frac{1}{2} \sigma^2
\right) t \right ), \forall t \in [0,T],
\end{equation}
where $S_0$ is the stock price at time zero, $T$ is the time to
expiration of the option.

The Black-Scholes price for a European call option  is
\begin{equation}\label{bsp}
C=S_0 N(d_1)-K e^{-rT} N(d_2),
\end{equation}
where $K$ is the strike price of the option, $N(\cdot)$ is the
cumulated distribution function of standard  normal distribution,
$r$ is the compounded riskless interest rate, and
\begin{equation}
d_1   =   \frac{\log \left( \frac{S_0}{K} \right) + \left(
r+\frac{1}{2} \sigma^2 \right) T  }
                {\sigma \sqrt{T}} , \quad
d_2   =   \frac{\log \left( \frac{S_0}{K} \right) + \left(
r-\frac{1}{2} \sigma^2 \right) T }
                {\sigma \sqrt{T}}
    =  d_1 - \sigma \sqrt{T}.
\end{equation}

The risk neural price of option is determined by those market
factors, except the growth rate of stock price. Therefore most of
the empirical studies do not use growth rate of stock price to
explain the biases of BS formula.

\subsection{Probabilities of Positive Returns}
It is the physical measure of the stock market which determines the
dynamics of the stock price, therefore the true probability of
positive return should be calculated under the physical measure,
instead of the artificial martingale measure. Let $C$ denote  the
price of a European call option at time zero, and $Pr$ denote  the
physical measure, then the probability of positive return is
\begin{equation}
     p(C)=Pr \left\{  \left( S_T-K \right)^{+} - C e^{rT} \ge 0
     \right\},
\end{equation}
where $ \left( S_T-K \right)^{+} =\max \left\{0,  S_T-K \right\} $.
We have the following result.
\begin{proposition}\label{prop1}
In the Black-Scholes model, when a European call option matures, the
probability for a holder to achieve positive return  is
\begin{eqnarray}\label{pequation}
p(C)
  &= &
  N(e_1)\cdot N(e_2),
\end{eqnarray}
where
\begin{equation}
e_1 = \frac{\log \left(\frac{S_0}{K}\right)+ \left( \mu -\frac{1}{2}
\sigma^2
          \right)T}{\sigma \sqrt{T}}, \quad
e_2 = \frac{\log\left(\frac{S_0}{K+C e^{rT}}\right)+ \left(\mu
-\frac{1}{2} \sigma^2
          \right)T}{\sigma \sqrt{T}},
\end{equation}
$N(\cdot)$ is the cumulated distribution function of standard
normal distribution, $C$ is the price of the call option, $S_0$ is
the stock price at time zero, $K$ is the strike price, $\mu$ is the
growth rate of stock price, $\sigma$ is the volatility, $r$ is the
riskless interest rate,
 and $T$ is the time to expiration.
\end{proposition}
\begin{proof}
See Appendix A.
\end{proof}

\begin{proposition}
The probability of positive return is a decreasing function of $C$,
$K$, $\sigma$ and $r$ respectively, and is an increasing function of
$\mu$ and $S_0$ respectively, when the other  factors are held
constant.
\end{proposition}
\begin{proof}
The proof is directly from the monotonicity of distribution
function.
\end{proof}

The probabilities of positive returns for different options will be
observed by the investors from the historical data, and the holders
will try to buy the options with probabilities beyond their
individual requirements, and wouldn't buy the options with lower
probabilities. The option price will be adjusted according to their
probabilities of positive returns to clear the market. The
equilibrium price of an option is a balance of probabilities for
both parties in the transaction.

\subsection{Probabilities and Biases of BS Formula}
As the probability of positive return is determined by the market
factors, we try to explain the observed biases of BS formula from
the compositions of market factors. The probabilities of positive
returns are calculated with respect to the BS values, and the
compositions of market factors whose probabilities of positive
returns are beyond $50\%$ are plotted in Figure \ref{urstk}.

\begin{figure}
\centering{
  \includegraphics[width=9cm]{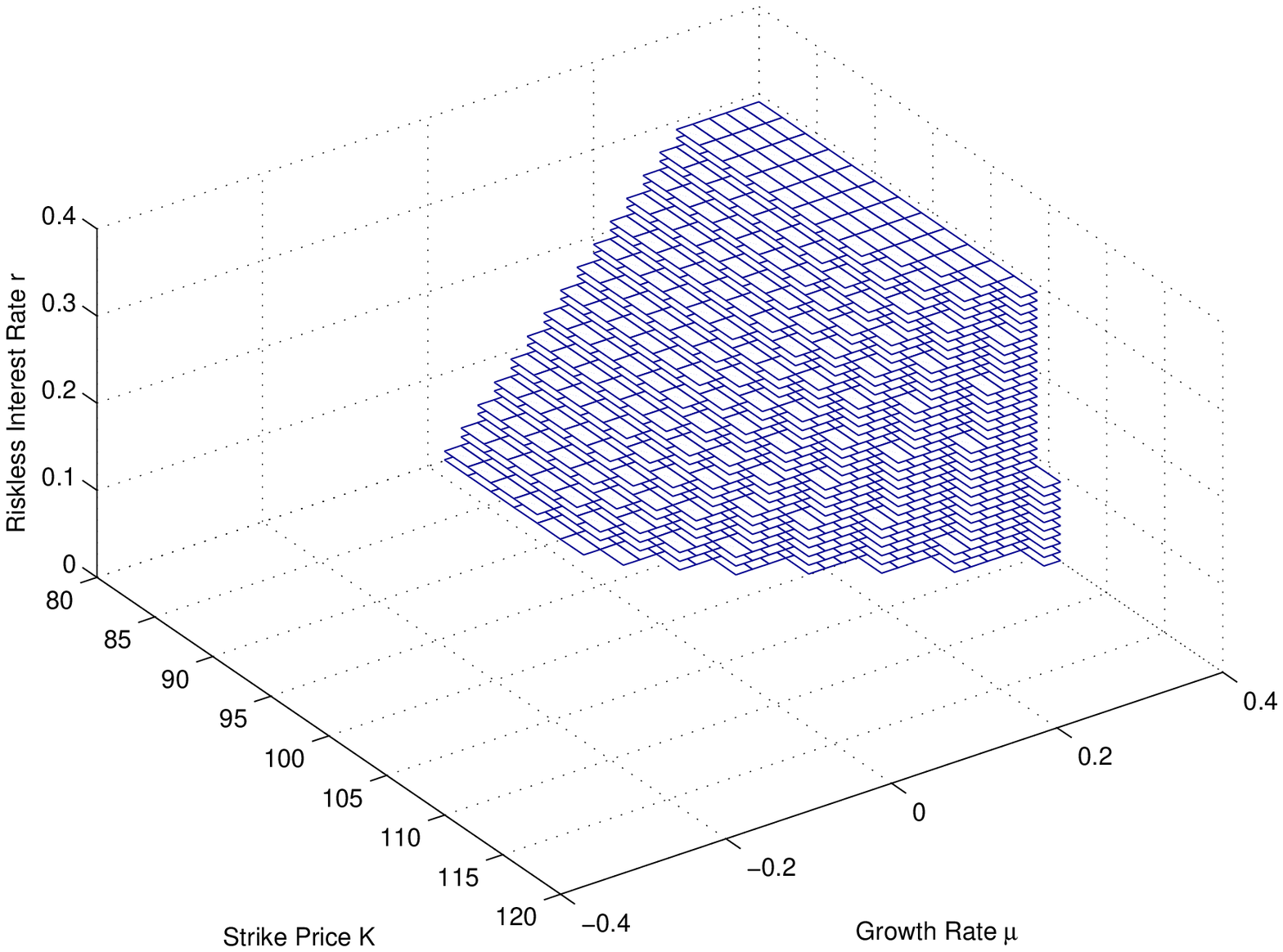}
\includegraphics[width=9cm]{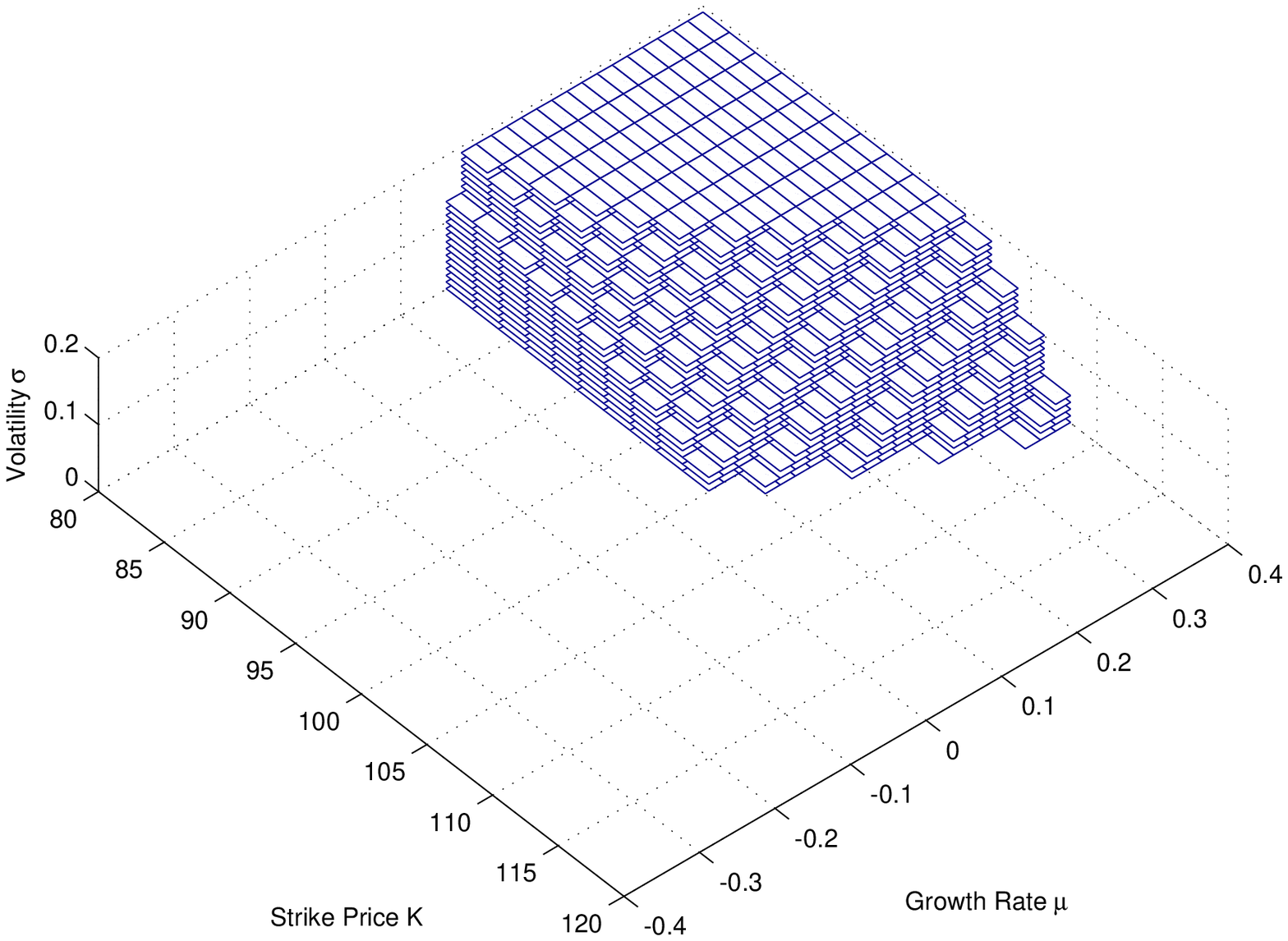}
  \includegraphics[width=9cm]{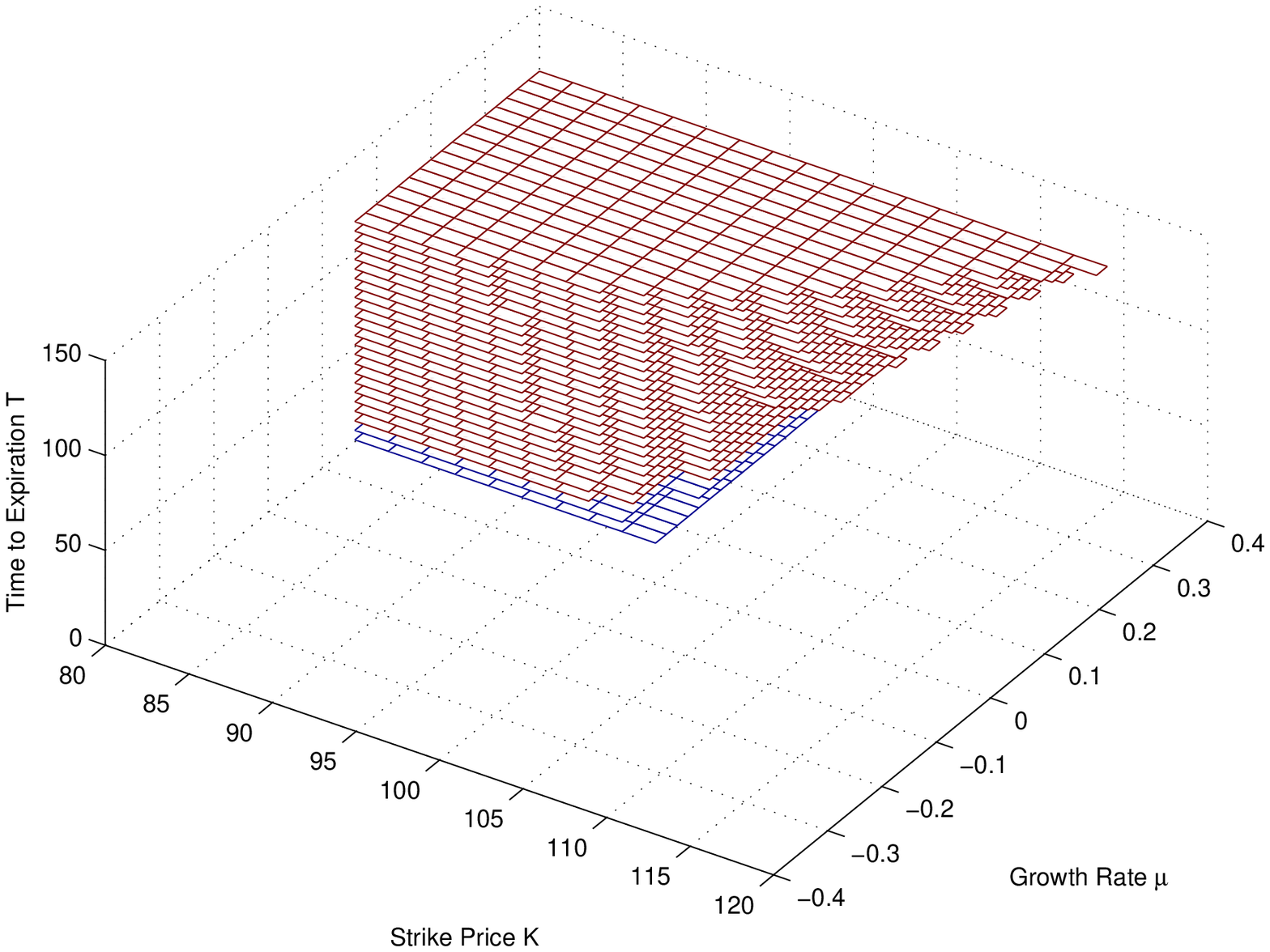}
  \caption[Compositions of $\mu$, $K$ $r$, $\sigma$ $T$ with probabilities larger than $50\%$.]
 {Compositions of $\mu$, $K$ $r$, $\sigma$, $T$ with probabilities larger than $50\%$.}\label{urstk}}
\end{figure}

Set the market factors as $S_0=100$, $K/S_0 \in [0.8,1.2]$, $\mu \in
[-0.4,0.4]$ per year. The compositions of $\left(K, \mu, r \right)$
are computed with $\sigma=0.1$ per year, $T=60$ days, $r \in
[0.001,0.3]$ per year. The compositions of $\left(K, \mu, \sigma
\right)$ are calculated with $r=0.05$ per year, $T=60$ days, $\sigma
\in [0.001,0.2]$ per year. And the compositions of $\left(K, \mu, T
\right)$ are calculated with $r=0.05$ per year, $\sigma = 0.1$ per
year, $T \in [1,120] $ days.

From figure (\ref{urstk}),  we can see that the probabilities of
positive returns must be less than $50\%$, when the growth rates of
stock prices are negative. The probabilities are on average beyond
$50\%$ for in-the-money options,  and below $50\%$ for
out-of-the-money options, when the growth rates of stock prices are
positive. The option holders will try to buy the options with
probabilities beyond $50\%$, and not to buy the options with
probabilities below $50\%$, which makes the market prices are on
average larger than the BS values for in-the-money options, and less
than the BS values for out-of-the-money options.

It is interesting to find that these theoretical predictions are
consistent with the empirical observations of
\cite{macbeth_merville_1979}, where the market prices are on average
lager than the BS values for in-the-money options, and less than the
BS values for out-of-the-money options.  Therefore the observed
anomalies of BS values may be correlated with the growth rates  of
stock prices, and the probabilities of positive returns should be
taken into account in derivative valuation.

\section{Probability of Positive Return and  Option Price}

When there are transaction costs or other frictions,
\cite{perrakis_ryan_1984,levy_1985,ritchken_1985} show that the
no-arbitrage conditions only place bounds on option prices,
therefore the prices of options should be determined through
equilibrium methods, rather than no-arbitrage models
\cite{longstaff_1995}.  The financial market without friction has
been discussed in the previous section, and the numerical
investigations indicate that the BS values should be adjusted
according to their probabilities of positive returns, therefore BS
values are not the equilibrium prices of options.

An equilibrium argument can be applied to determine the fair price
of option through the probability of positive return. If the
probability of positive return is larger than the equilibrium level,
the investors would try to buy this option until the option price
increases to eliminate the extra possibility of positive return. And
the investors would not buy the options with probabilities below the
equilibrium level, until the option price decrease enough to
compensate the risk of negative return.  On the other hand, the
requirements of probabilities to achieve positive returns may be
different among the investors according to their attitudes to risk.

The option price can be deduced from (\ref{pequation}), if the
equilibrium probability $p$ is determined previously, and we have
the following proposition.
\begin{proposition}\label{propprice}
In the Black-Scholes market, if the equilibrium probability of
positive return is $p$, and $p<N(e_1)$, the equilibrium price of a
European call option  is
\begin{equation}\label{call}
C(p)=S_0 e^{-\sigma \sqrt{ T}  N^{-1} ( \frac{p }{ N(e_1)} )
              +(\mu-r-\frac{1}{2}\sigma^2)T
} - e^{-rT}K.
\end{equation}
Otherwise, the required probability of positive return is beyond the
the probability of the option to be exercised, therefore $C(p)$ does
not exist (NaN). Where $N^{-1}(\cdot)$ is the inverse function of
the standard normal distribution function.
\end{proposition}

\begin{remark}
From equation (\ref{call}), we can find that the individual attitude
to risk
 is important to option
valuation. The true probability of an option to be exercised is
$N(e_1)$, and those investors, whose individual requirements for the
probabilities of positive returns are beyond $N(e_1)$, would not buy
this option. As  market prices should lie in the no-arbitrage
bounds, the option price $C(p)$ will be modified as $ \min \left\{0,
S_0 - e^{-rT}K \right\} \le \tilde{C}(p) \le S_0$.
\end{remark}

\begin{proposition}
In the Black-Scholes model, the equilibrium option price $C(p)$ is a
decreasing function with respect to $p$.
\end{proposition}
\begin{proof}
It is directly from equation (\ref{call}).
\end{proof}

This proposition says that the investors, whose requirements for
probabilities of positive returns are larger than $p$, would not buy
the options with prices beyond $C(p)$. On the other hand, these
investors will buy the options with prices below $C(p)$.

\begin{sidewaystable}
\tabcolsep=0.14cm \small \centering
\begin{tabular}{z{3}z{3}z{3}z{3}z{3}z{3}z{3}z{3}z{3}rrrrrrrrr}
\toprule
  $K$  &   80.00   &   82.00   &   84.00   &   86.00   &   88.00   &   90.00   &   92.00   &   94.00   &   96.00   &   98.00   &   100.00  &   102.00  &   104.00  &   106.00  &   108.00  &   110.00 & 112.00 \\
\midrule
$BS$    &   21.85   &   19.90   &   17.95   &   15.99   &   14.04   &   12.09   &   10.15   &   8.25    &   6.43    &   4.74    &   3.29    &   2.11    &   1.25    &   0.68    &   0.34    &   0.15    &   0.06    \\
\midrule
-0.25   &   21.85   &   19.90   &   17.95   &   15.99   &   14.04   &   12.08   &   10.13   &   8.18    &   6.22    &   4.27    &   $NaN$   &   $NaN$   &   $NaN$   &   $NaN$   &   $NaN$   &   $NaN$   &   $NaN$   \\
-0.23   &   21.85   &   19.90   &   17.95   &   15.99   &   14.04   &   12.08   &   10.13   &   8.18    &   6.22    &   4.27    &   $NaN$   &   $NaN$   &   $NaN$   &   $NaN$   &   $NaN$   &   $NaN$   &   $NaN$   \\
-0.21   &   21.85   &   19.90   &   17.95   &   15.99   &   14.04   &   12.08   &   10.13   &   8.18    &   6.22    &   4.27    &   $NaN$   &   $NaN$   &   $NaN$   &   $NaN$   &   $NaN$   &   $NaN$   &   $NaN$   \\
-0.19   &   21.85   &   19.90   &   17.95   &   15.99   &   14.04   &   12.08   &   10.13   &   8.18    &   6.22    &   4.27    &   $NaN$   &   $NaN$   &   $NaN$   &   $NaN$   &   $NaN$   &   $NaN$   &   $NaN$   \\
-0.17   &   21.85   &   19.90   &   17.95   &   15.99   &   14.04   &   12.08   &   10.13   &   8.18    &   6.22    &   4.27    &   $NaN$   &   $NaN$   &   $NaN$   &   $NaN$   &   $NaN$   &   $NaN$   &   $NaN$   \\
-0.15   &   21.85   &   19.90   &   17.95   &   15.99   &   14.04   &   12.08   &   10.13   &   8.18    &   6.22    &   4.27    &   2.32    &   $NaN$   &   $NaN$   &   $NaN$   &   $NaN$   &   $NaN$   &   $NaN$   \\
-0.13   &   21.85   &   19.90   &   17.95   &   15.99   &   14.04   &   12.08   &   10.13   &   8.18    &   6.22    &   4.27    &   2.32    &   $NaN$   &   $NaN$   &   $NaN$   &   $NaN$   &   $NaN$   &   $NaN$   \\
-0.11   &   21.85   &   19.90   &   17.95   &   15.99   &   14.04   &   12.08   &   10.13   &   8.18    &   6.22    &   4.27    &   2.32    &   $NaN$   &   $NaN$   &   $NaN$   &   $NaN$   &   $NaN$   &   $NaN$   \\
-0.09   &   21.85   &   19.90   &   17.95   &   15.99   &   14.04   &   12.08   &   10.13   &   8.18    &   6.22    &   4.27    &   2.32    &   $NaN$   &   $NaN$   &   $NaN$   &   $NaN$   &   $NaN$   &   $NaN$   \\
-0.07   &   21.85   &   19.90   &   17.95   &   15.99   &   14.04   &   12.08   &   10.13   &   8.18    &   6.22    &   4.27    &   2.32    &   0.36    &   $NaN$   &   $NaN$   &   $NaN$   &   $NaN$   &   $NaN$   \\
-0.05   &   22.30   &   20.34   &   18.39   &   16.43   &   14.45   &   12.43   &   10.31   &   8.18    &   6.22    &   4.27    &   2.32    &   0.36    &   $NaN$   &   $NaN$   &   $NaN$   &   $NaN$   &   $NaN$   \\
-0.03   &   22.77   &   20.81   &   18.86   &   16.90   &   14.93   &   12.92   &   10.83   &   8.58    &   6.22    &   4.27    &   2.32    &   0.36    &   $NaN$   &   $NaN$   &   $NaN$   &   $NaN$   &   $NaN$   \\
-0.01   &   23.24   &   21.29   &   19.33   &   17.38   &   15.41   &   13.41   &   11.34   &   9.13    &   6.70    &   4.27    &   2.32    &   0.36    &   $NaN$   &   $NaN$   &   $NaN$   &   $NaN$   &   $NaN$   \\
0.01    &   23.72   &   21.76   &   19.81   &   17.85   &   15.89   &   13.90   &   11.85   &   9.68    &   7.30    &   4.61    &   2.32    &   0.36    &   0.00    &   $NaN$   &   $NaN$   &   $NaN$   &   $NaN$   \\
0.03    &   24.20   &   22.24   &   20.29   &   18.33   &   16.37   &   14.39   &   12.36   &   10.22   &   7.90    &   5.29    &   2.32    &   0.36    &   0.00    &   $NaN$   &   $NaN$   &   $NaN$   &   $NaN$   \\
0.05    &   24.68   &   22.72   &   20.77   &   18.82   &   16.86   &   14.88   &   12.86   &   10.75   &   8.48    &   5.94    &   3.01    &   0.36    &   0.00    &   $NaN$   &   $NaN$   &   $NaN$   &   $NaN$   \\
0.07    &   25.16   &   23.21   &   21.25   &   19.30   &   17.34   &   15.37   &   13.36   &   11.28   &   9.05    &   6.58    &   3.76    &   0.37    &   0.00    &   $NaN$   &   $NaN$   &   $NaN$   &   $NaN$   \\
0.09    &   25.65   &   23.69   &   21.74   &   19.78   &   17.83   &   15.86   &   13.86   &   11.80   &   9.61    &   7.21    &   4.48    &   1.25    &   0.00    &   0.00    &   $NaN$   &   $NaN$   &   $NaN$   \\
0.11    &   26.13   &   24.18   &   22.23   &   20.27   &   18.32   &   16.35   &   14.37   &   12.32   &   10.17   &   7.82    &   5.18    &   2.09    &   0.00    &   0.00    &   $NaN$   &   $NaN$   &   $NaN$   \\
0.13    &   26.62   &   24.67   &   22.72   &   20.76   &   18.81   &   16.85   &   14.87   &   12.84   &   10.71   &   8.42    &   5.86    &   2.90    &   0.00    &   0.00    &   $NaN$   &   $NaN$   &   $NaN$   \\
0.15    &   27.12   &   25.16   &   23.21   &   21.25   &   19.30   &   17.34   &   15.37   &   13.35   &   11.25   &   9.01    &   6.52    &   3.67    &   0.24    &   0.00    &   $NaN$   &   $NaN$   &   $NaN$   \\
0.17    &   27.61   &   25.66   &   23.70   &   21.75   &   19.80   &   17.84   &   15.87   &   13.87   &   11.79   &   9.58    &   7.16    &   4.41    &   1.15    &   0.00    &   0.00    &   $NaN$   &   $NaN$   \\
0.19    &   28.11   &   26.15   &   24.20   &   22.25   &   20.29   &   18.34   &   16.37   &   14.38   &   12.32   &   10.15   &   7.79    &   5.13    &   2.01    &   0.00    &   0.00    &   $NaN$   &   $NaN$   \\
0.21    &   28.61   &   26.65   &   24.70   &   22.75   &   20.79   &   18.84   &   16.87   &   14.89   &   12.85   &   10.71   &   8.40    &   5.82    &   2.84    &   0.00    &   0.00    &   $NaN$   &   $NaN$   \\
0.23    &   29.11   &   27.15   &   25.20   &   23.25   &   21.29   &   19.34   &   17.38   &   15.40   &   13.38   &   11.27   &   9.01    &   6.50    &   3.63    &   0.18    &   0.00    &   $NaN$   &   $NaN$   \\
0.25    &   29.61   &   27.66   &   25.71   &   23.75   &   21.80   &   19.84   &   17.88   &   15.91   &   13.90   &   11.82   &   9.60    &   7.16    &   4.39    &   1.11    &   0.00    &   0.00    &   $NaN$   \\
\bottomrule
\end{tabular}
\caption[BS values and $C(20\%)$ values for different growth rates.]
{BS values and $C(20\%)$ values for different growth rates. The
strike prices are listed in the first row, and the second row are
the BS values corresponding to those strike prices. The growth rates
of stock price are listed in the first column, varying from $-0.25$
to $0.25$ per year. And the $C(20\%)$ values are reported according
to the growth rates and strike prices respectively. NaN means that
there is no possibility to achieve positive returns with
probabilities larger than $20\%$ under the specified
market.}\label{p0.2tab}
\end{sidewaystable}

\begin{sidewaystable}
\tabcolsep=0.17cm \small \centering
\begin{tabular}{z{3}z{3}z{3}z{3}z{3}z{3}z{3}z{3}z{3}rrrrrrrr}
\toprule
  $K$  &   80.00   &   82.00   &   84.00   &   86.00   &   88.00   &   90.00   &   92.00   &   94.00   &   96.00   &   98.00   &   100.00  &   102.00  &   104.00  &   106.00  &   108.00  &   110.00  \\
\midrule
$BS$  &   21.85   &   19.90   &   17.95   &   15.99   &   14.04   &   12.09   &   10.15   &   8.25    &   6.43    &   4.74    &   3.29    &   2.11    &   1.25    &   0.68    &   0.34    &   0.15    \\
\midrule
-0.25   &   21.85   &   19.90   &   17.95   &   15.99   &   14.04   &   12.08   &   10.13   &   8.18    &   $NaN$   &   $NaN$   &   $NaN$   &   $NaN$   &   $NaN$   &   $NaN$   &   $NaN$   &   $NaN$   \\
-0.23   &   21.85   &   19.90   &   17.95   &   15.99   &   14.04   &   12.08   &   10.13   &   8.18    &   $NaN$   &   $NaN$   &   $NaN$   &   $NaN$   &   $NaN$   &   $NaN$   &   $NaN$   &   $NaN$   \\
-0.21   &   21.85   &   19.90   &   17.95   &   15.99   &   14.04   &   12.08   &   10.13   &   8.18    &   $NaN$   &   $NaN$   &   $NaN$   &   $NaN$   &   $NaN$   &   $NaN$   &   $NaN$   &   $NaN$   \\
-0.19   &   21.85   &   19.90   &   17.95   &   15.99   &   14.04   &   12.08   &   10.13   &   8.18    &   $NaN$   &   $NaN$   &   $NaN$   &   $NaN$   &   $NaN$   &   $NaN$   &   $NaN$   &   $NaN$   \\
-0.17   &   21.85   &   19.90   &   17.95   &   15.99   &   14.04   &   12.08   &   10.13   &   8.18    &   $NaN$   &   $NaN$   &   $NaN$   &   $NaN$   &   $NaN$   &   $NaN$   &   $NaN$   &   $NaN$   \\
-0.15   &   21.85   &   19.90   &   17.95   &   15.99   &   14.04   &   12.08   &   10.13   &   8.18    &   6.22    &   $NaN$   &   $NaN$   &   $NaN$   &   $NaN$   &   $NaN$   &   $NaN$   &   $NaN$   \\
-0.13   &   21.85   &   19.90   &   17.95   &   15.99   &   14.04   &   12.08   &   10.13   &   8.18    &   6.22    &   $NaN$   &   $NaN$   &   $NaN$   &   $NaN$   &   $NaN$   &   $NaN$   &   $NaN$   \\
-0.11   &   21.85   &   19.90   &   17.95   &   15.99   &   14.04   &   12.08   &   10.13   &   8.18    &   6.22    &   $NaN$   &   $NaN$   &   $NaN$   &   $NaN$   &   $NaN$   &   $NaN$   &   $NaN$   \\
-0.09   &   21.85   &   19.90   &   17.95   &   15.99   &   14.04   &   12.08   &   10.13   &   8.18    &   6.22    &   $NaN$   &   $NaN$   &   $NaN$   &   $NaN$   &   $NaN$   &   $NaN$   &   $NaN$   \\
-0.07   &   21.85   &   19.90   &   17.95   &   15.99   &   14.04   &   12.08   &   10.13   &   8.18    &   6.22    &   4.27    &   $NaN$   &   $NaN$   &   $NaN$   &   $NaN$   &   $NaN$   &   $NaN$   \\
-0.05   &   21.85   &   19.90   &   17.95   &   15.99   &   14.04   &   12.08   &   10.13   &   8.18    &   6.22    &   4.27    &   $NaN$   &   $NaN$   &   $NaN$   &   $NaN$   &   $NaN$   &   $NaN$   \\
-0.03   &   21.85   &   19.90   &   17.95   &   15.99   &   14.04   &   12.08   &   10.13   &   8.18    &   6.22    &   4.27    &   $NaN$   &   $NaN$   &   $NaN$   &   $NaN$   &   $NaN$   &   $NaN$   \\
-0.01   &   21.85   &   19.90   &   17.95   &   15.99   &   14.04   &   12.08   &   10.13   &   8.18    &   6.22    &   4.27    &   $NaN$   &   $NaN$   &   $NaN$   &   $NaN$   &   $NaN$   &   $NaN$   \\
0.01    &   21.85   &   19.90   &   17.95   &   15.99   &   14.04   &   12.08   &   10.13   &   8.18    &   6.22    &   4.27    &   2.32    &   $NaN$   &   $NaN$   &   $NaN$   &   $NaN$   &   $NaN$   \\
0.03    &   21.85   &   19.90   &   17.95   &   15.99   &   14.04   &   12.08   &   10.13   &   8.18    &   6.22    &   4.27    &   2.32    &   $NaN$   &   $NaN$   &   $NaN$   &   $NaN$   &   $NaN$   \\
0.05    &   21.85   &   19.90   &   17.95   &   15.99   &   14.04   &   12.08   &   10.13   &   8.18    &   6.22    &   4.27    &   2.32    &   $NaN$   &   $NaN$   &   $NaN$   &   $NaN$   &   $NaN$   \\
0.07    &   21.85   &   19.90   &   17.95   &   15.99   &   14.04   &   12.08   &   10.13   &   8.18    &   6.22    &   4.27    &   2.32    &   $NaN$   &   $NaN$   &   $NaN$   &   $NaN$   &   $NaN$   \\
0.09    &   21.85   &   19.90   &   17.95   &   15.99   &   14.04   &   12.08   &   10.13   &   8.18    &   6.22    &   4.27    &   2.32    &   0.36    &   $NaN$   &   $NaN$   &   $NaN$   &   $NaN$   \\
0.11    &   21.97   &   20.02   &   18.06   &   16.11   &   14.15   &   12.18   &   10.17   &   8.18    &   6.22    &   4.27    &   2.32    &   0.36    &   $NaN$   &   $NaN$   &   $NaN$   &   $NaN$   \\
0.13    &   22.44   &   20.49   &   18.53   &   16.58   &   14.62   &   12.66   &   10.66   &   8.58    &   6.33    &   4.27    &   2.32    &   0.36    &   $NaN$   &   $NaN$   &   $NaN$   &   $NaN$   \\
0.15    &   22.91   &   20.96   &   19.01   &   17.05   &   15.10   &   13.13   &   11.14   &   9.09    &   6.89    &   4.40    &   2.32    &   0.36    &   $NaN$   &   $NaN$   &   $NaN$   &   $NaN$   \\
0.17    &   23.39   &   21.43   &   19.48   &   17.53   &   15.57   &   13.61   &   11.63   &   9.59    &   7.43    &   5.03    &   2.32    &   0.36    &   $NaN$   &   $NaN$   &   $NaN$   &   $NaN$   \\
0.19    &   23.87   &   21.91   &   19.96   &   18.00   &   16.05   &   14.09   &   12.12   &   10.10   &   7.97    &   5.64    &   2.91    &   0.36    &   0.00    &   $NaN$   &   $NaN$   &   $NaN$   \\
0.21    &   24.34   &   22.39   &   20.44   &   18.48   &   16.53   &   14.57   &   12.60   &   10.60   &   8.50    &   6.23    &   3.61    &   0.36    &   0.00    &   $NaN$   &   $NaN$   &   $NaN$   \\
0.23    &   24.83   &   22.87   &   20.92   &   18.97   &   17.01   &   15.05   &   13.09   &   11.09   &   9.03    &   6.80    &   4.29    &   1.23    &   0.00    &   $NaN$   &   $NaN$   &   $NaN$   \\
0.25    &   25.31   &   23.36   &   21.40   &   19.45   &   17.50   &   15.54   &   13.58   &   11.59   &   9.54    &   7.37    &   4.93    &   2.04    &   0.00    &   $NaN$   &   $NaN$    &   $NaN$   \\
\bottomrule
\end{tabular}
\caption[BS values and $C(50\%)$ values for different growth rates.]
{BS values and $C(50\%)$ values for different growth rates. The
strike prices are listed in the first row, and the second row are
the BS values corresponding to those strike prices. The growth rates
of stock price are listed in the first column, varying from $-0.25$
to $0.25$ per year. And the $C(50\%)$ values are reported according
to the growth rates and strike prices respectively. NaN means that
there is no possibility to achieve positive returns with
probabilities larger than $50\%$ under the specified
market.}\label{p0.5tab}
\end{sidewaystable}

\section{Comparisons of the Two Option Prices}
\subsection{Growth Rates and Biases of BS Values}
In order to investigate the performance of the proposed method, set
the market factors as $S_0=100$, $r=0.05$ per year, $\sigma=0.1$ per
year, $T=60$ days, and calculate the call option prices  for $K/S
\in [0.8,1.15]$, $\mu \in [-0.25,0.25]$ per year, using both BS
formula and the proposed method. The results are reported in table
(\ref{p0.2tab}) and table (\ref{p0.5tab}) with respect to $p=20\%$
and $p=50\%$.

From table (\ref{p0.2tab}), we can find that the BS values are
systematically larger than or equal to those $C(20\%)$ values for
all of the strike prices, when the growth rates are between $-0.25$
to $-0.07$. For those deep in-the-money options, the BS values are
gradually less than the $C(20\%)$ values, when the growth rates are
larger than $-0.05$. For those slightly in-the-money options, the BS
values are most likely larger than those $C(20\%)$ values, when the
growth rates are between $-0.05$ to $-0.01$. When the growth rates
are larger than $0.01$,  the BS values of in-the-money options are
all less than the $C(20\%)$ values. For those deep out-of-the-money
options, the BS values are systematically larger than those
$C(20\%)$ values. For those slightly out-of-the-money options, the
BS values are less than the $C(20\%)$ values, when the growth rates
are larger than $0.13$.

A similar phenomena  is observed in table (\ref{p0.5tab}) with
$p=50\%$. We can conclude that the values of proposed method are on
average larger than the BS values for in-the-money options, and less
than the BS values for out-of-the-money options in our numerical
investigations. It is interesting to find that those theoretical
observations are consistent with the empirical results of
\cite{macbeth_merville_1979}, where the market prices of options are
on average larger than the BS values for in-the-money options, and
less than the BS values for out-of-the-money options.

\subsection{Implied Volatility}
\begin{figure}
\centering{
  \includegraphics[width=10cm]{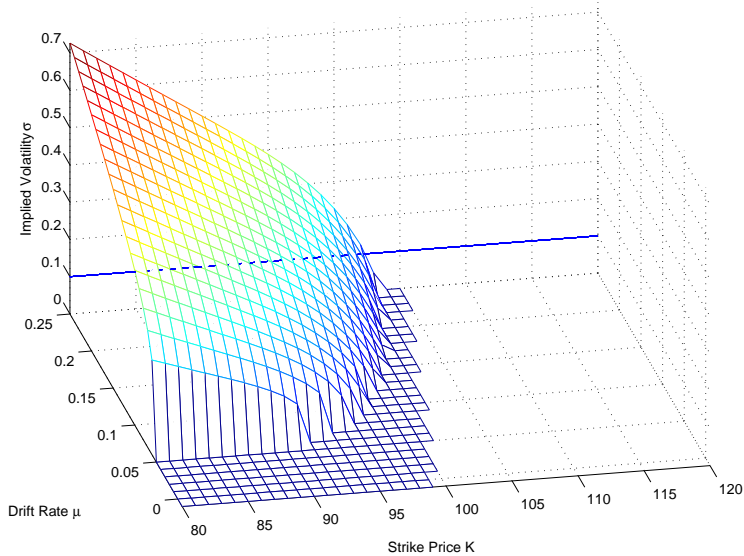}
  \caption{The surface of implied volatilities and the reference line.}\label{imv}
}
\end{figure}

As the influence of growth rate has been taken into account in the
values of $C(p)$, it is possible to investigate the implied
volatilities for  different growth rates with respect to BS formula.
Set the market factors as $S_0=100$, $T=60$ days, $\sigma=0.1$ per
year, $r=0.05$ per year, $p=0.5$,
 $K/S \in [0.8,1.2]$, $\mu \in  [-0.1,0.25]$ per year.
The surface of implied volatilities for different $\mu$ is plotted
in figure (\ref{imv}).

From figure (\ref{imv}),  we can see that there is significant
inconsistency of implied volatilities using the   $C(50\%)$ values
as  market prices. The shape of the implied volatilities is similar
to the empirical results of many authors, such as
\cite{kuwahara_marsh_2000}. Therefore the growth rate  of stock
price should be used to explain the biases of BS formula, and the
probabilities of positive returns should be taken into account in
derivative valuation.

\section{Conclusions and Discussions}
The biases of Black-Scholes formula have been discussed widely for a
long time, and many market factors have been used to explain the
observed anomalies of BS values, such as strike price, volatility,
time to expiration, bid-ask spread, trading volume, option open
interest, transaction cost, etc. Although the growth rate is an
important factor to describe the dynamics of stock prices, it is
rarely discussed in the literature for the biases of BS formula.
Most of the reason is that the physical growth rate of stock price
is eliminated on the procedure of risk neural valuation, and the
resulted option price is not correlated with the growth rate. The
influence of growth rate on the biases of BS formula is investigated
under the BS model, and it is found that the larger the growth rate,
the more possible the option to be exercised with  positive returns.
Therefore the BS values are not the equilibrium prices, and the
market prices need to be adjusted from their BS values to the
equilibrium prices, according to the probabilities of positive
returns.

An alternative valuation method for European call option is proposed
in this paper. The probability of positive return is used to
identify the risk of holding an option, and the equilibrium price of
option will balance the requirements for probabilities of positive
returns between the both parties in the transaction. The performance
of  BS formula and the proposed method are compared under the BS
model, and it is found that the values of proposed method are on
average larger than the  BS values for in-the-money options, and
less than the  BS values for out-of-the-money options. These
theoretical phenomenon are similar to the empirical results reported
by \cite{macbeth_merville_1979}, if we take the values of proposed
method as the market prices.

The values of proposed method are also used to calculate the implied
volatilities from BS formula. It is found that the shape of implied
volatilities is similar to the empirical observations, such as
\cite{kuwahara_marsh_2000}. From the biases of BS values and the
shape of implied volatilities in our numerical investigations, we
can conclude that there must be some similarity between the market
prices and the values of proposed method, therefore the proposed
valuation method may have some merit.

As the focus of this paper is to discuss the possibility to explain
the observed anomalies of BS formula by the growth rate of stock
price theoretically, the empirical performance of the proposed
valuation method is not investigated in detail. Future research
should investigate the empirical performance of the proposed method,
and pay more attention  to  the mechanism of market equilibrium,
when the probabilities of positive returns are taken into account in
the investment decisions.

\small {\section*{Acknowledgment}} The authors want to thank Prof.
Zhiyuan Huang, Prof. Chujin Li, Dr. Xu Chen for helpful discussions.
And this work is supported by the starting foundation of Chongqing
University, No. 0903005104882.

\appendix

\section{The proof of Proposition \ref{prop1}}

\begin{proof}
The probability of positive return can be rewritten as
\begin{eqnarray}\label{pstpr}
p(C)
  & = &
  Pr  \left\{  \left( S_T-K \right)^{+} - C e^{rT} \ge 0   | S_T-K \ge 0
       \right\} \cdot  Pr  \left\{ S_T-K \ge 0
       \right\}\nonumber \\
      &  & +
       Pr  \left\{  \left( S_T-K \right)^{+} - C e^{rT} \ge 0   | S_T-K < 0
     \right\} \cdot        Pr  \left\{  S_T-K < 0
     \right\}\nonumber \\
  & = & Pr  \left\{ S_T \ge K
       \right\}\cdot
  Pr  \left\{   S_T \ge K + C e^{rT}   \right\}.
\end{eqnarray}
On the other hand, we have
\begin{eqnarray}\label{e1}
 Pr  \left\{ S_T \ge K
       \right\}
& = &  Pr  \left\{ S_0 exp\left(  \sigma W_T + \left(\mu
-\frac{1}{2} \sigma^2  \right)T \right) \ge K
       \right\}\nonumber \\
 &       = & Pr  \left\{ \frac{-\sigma W_T}{\sigma \sqrt{T}}  \le
 \frac{\log \left( \frac{S_0}{K} \right)+ (\mu -\frac{1}{2} \sigma^2  )T}{\sigma \sqrt{T}}
       \right\}\nonumber \\
  &     = &
       N(e_1).
\end{eqnarray}
Following a similar procedure, we have
\begin{equation}\label{e2}
Pr  \left\{   S_T \ge K + C e^{rT}   \right\}
  =N(e_2),
\end{equation}
and the desired result will follow.
\end{proof}


\section{The proof of Proposition \ref{propprice}}
\begin{proof} When $p > N \left(e_1 \right)$ is held, the inverse function
$N^{-1}\left( \frac{p}{N \left(e_1 \right)} \right)$ dose not exist.
 When $p<N \left(e_1 \right)$ is held, from $p=N(e_1)N(e_2)$, we have
$N \left(e_2 \right) = \frac{p}{N \left(e_1 \right)}$, and it
follows that
\begin{equation}\label{prob}
e_2  =N^{-1} \left( \frac{p}{N \left(e_1 \right)} \right).
\end{equation}
On the other hand, we have
\begin{equation}\label{probb}
e_2 = \frac{\log\left(\frac{S_0}{K+C e^{rT}}\right)+ \left(\mu
-\frac{1}{2} \sigma^2
          \right)T}{\sigma \sqrt{T}}.
\end{equation}
Substituting (\ref{probb}) into (\ref{prob}), the desired result
follows.
\end{proof}




\end{document}